\documentclass[10 pt, conference]{ieeeconf}
\IEEEoverridecommandlockouts

\title{\LARGE \bf
Data-driven distributed control:\\ Virtual reference feedback tuning in dynamic networks
}

\author{Tom R.V. Steentjes, Mircea Lazar, Paul M.J. Van den Hof
\thanks{T.R.V. Steentjes, M. Lazar and P.M.J. Van den Hof are with the Department of Electrical Engineering, Eindhoven University of Technology, The Netherlands. E-mails: \texttt{\{t.r.v.steentjes, m.lazar, p.m.j.vandenhof\}@tue.nl}}
\thanks{This work is supported by the European Research Council (ERC), Advanced Research Grant SYSDYNET, under the European Unions Horizon 2020 research and innovation programme (grant agreement No. 694504).}
}

\usepackage{amsmath}
\usepackage{amssymb}
\usepackage{mathrsfs}
\usepackage{graphicx}
\usepackage{theorem}
\usepackage{algorithm,algorithmic}
\usepackage{subcaption}
\usepackage{xcolor}
\usepackage{cite}

{\theorembodyfont{\slshape}\newtheorem{thm}{Theorem}[section]}
{\theorembodyfont{\slshape}\newtheorem{corollary}{Corollary}[section]}
{\theorembodyfont{\slshape}}
{\theorembodyfont{\upshape}}
{\theorembodyfont{\upshape}\newtheorem{problem}{Problem}[section]}
{\theorembodyfont{\upshape}}
{\theorembodyfont{\slshape}}
{\theorembodyfont{\upshape}\newtheorem{example}{Example}[section]}
{\theorembodyfont{\upshape}\newtheorem{assumption}{Assumption}[section]}

\begin{document}
\maketitle
\begin{abstract}
In this paper, the problem of synthesizing a \emph{distributed} controller from data is considered, with the objective to optimize a model-reference control criterion. We establish an explicit ideal distributed controller that solves the model-reference control problem for a structured reference model. On the basis of input-output data collected from the interconnected system, a virtual experiment setup is constructed which leads to a network identification problem. We formulate a prediction-error identification criterion that has the same global optimum as the model-reference criterion, when the controller class contains the ideal distributed controller. The developed distributed controller synthesis method is illustrated on an academic example network of nine subsystems and the influence of the controller interconnection structure on the achieved closed-loop performance is analyzed.
\end{abstract}
\section{Introduction}
Control of interconnected systems is a challenging problem. Firstly, due to the spatial distribution or dimensionality, which prohibits the use of centralized controller design and implementation. Moreover, for many practical control applications, such as smart grids, smart buildings or industrial processes, dynamical models are not readily available, while measurement data is available with increased ease \cite{reviewarticle2017}. A relevant question is how to directly exploit the available data for distributed controller synthesis, without using a model.

Indeed, most of the methods to design distributed controllers are based on a model of the interconnected system, e.g. distributed model predictive control \cite{christofides2013}, distributed $H_2$ \cite{chen2019} and $H_\infty$ \cite{langbortetal2004} control. From a model-based perspective, a logical procedure is the data-driven modelling of the interconnected system and subsequent synthesis of the distributed controller based on the obtained model. Performance-oriented data-driven modelling of lumped systems has been founded in the field of `identification for control' \cite{vandenhofetal1995}. The development of data-driven modelling of interconnected systems \cite{vandenhofetal2013} opens the way for control-oriented identification for distributed control.

For data-driven control, the step of modelling the interconnected system may be circumvented, however, and the design of the distributed controller could be directly performed on the basis of data. Several methods have been developed for data-driven controller design, see e.g. \cite{hou2013} for an overview. A common feature of these methods is that they are based on the model-reference paradigm, wherein a reference model describes the desired behavior of the closed-loop system \cite{bazanella2011}. Virtual reference feedback tuning (VRFT) \cite{campi2002} is a `one-shot' method in which a model-reference criterion is optimized using a single batch of measurement data. Recent developments of `one-shot' data-driven methods include multi-variable VRFT~\cite{campestrini2016}, optimal multi-variable controller identification \cite{huff2019} and asymptotically exact multi-variable controller tuning \cite{formentin2015}. The step to the distributed case is of significant interest, given the potential of data-driven methods for interconnected systems~\cite{reviewarticle2017}.

In this paper, following the model-reference paradigm, we specify the desired behavior for the closed-loop system in terms of a \emph{structured} reference model. The accompanying control problem is to find a distributed controller which minimizes the model reference criterion. With the introduction of an ideal distributed controller, we provide an analogy to the ideal controller for the standard model-reference problem~\cite{bazanella2011}. Then, with the extension of VRFT to a distributed setting, we are able to synthesize a data-driven controller through dynamic network identification~\cite{vandenhofetal2013}. The contribution to distributed control is the direct data-driven design, as the synthesis of distributed controllers is typically model based. Regarding data-driven control, the contribution is the synthesis of distributed data-driven controllers with a priori defined structure and identification of distributed controllers via network identification techniques.

\section{Preliminaries}
Consider an undirected graph $\mathcal{G}=(\mathcal{V},\mathcal{E})$ with vertex set $\mathcal{V}$ of cardinality $L$ and edge set $\mathcal{E}\subseteq \mathcal{V}\times \mathcal{V}$. The neighbour set of vertex $i\in \mathcal{V}$ is defined as $\mathcal{N}_i:=\{j\in \mathcal{V}\,|\, (i,j)\in \mathcal{E}\}$. To each vertex $i\in \mathcal{V}$, we associate a linear discrete-time system with dynamics
\begin{align*}
y_i(t)&= G_i(q)u_i(t)+\displaystyle{\sum_{j\in \mathcal{N}_i}} W_{ij}(q)s_{ij}(t),\\
o_{ij}(t)& = F_{ij}(q)y_i(t),\qquad j\in \mathcal{N}_i,
\end{align*}
with $G_i$, $W_{ij}$, $F_{ij}$ rational transfer functions, $q$ the forward shift defined as $qx(t)=x(t+1)$, $u_i:\mathbb{Z}\to \mathbb{R}$ is the control input, $y_i:\mathbb{Z}\to \mathbb{R}$ the output, and $o_{ij},s_{ij}:\mathbb{Z}\to \mathbb{R}$ are variables through which the systems at vertices $(i,j)\in \mathcal{E}$ are interconnected. The problem that we consider is that of reference tracking, i.e., for each system it is desired that the output $y_i$ tracks a reference signal $r_i$. The tracking error for system $i$ is defined as $e_i:=r_i-y_i$. By stacking all incoming and outgoing interconnection variables of system $i$ in vectors $s_i$ and $o_i$, that is $s_i:=\operatorname{col}_{j\in \mathcal{N}_i} s_{ij}$ and $o_i:=\operatorname{col}_{j\in \mathcal{N}_i} o_{ij}$, we arrive at the following description for system $i$, denoted $\mathcal{P}_i$:
\begin{align} \label{eq:pi}
\mathcal{P}_i:\left\{\begin{array}{lll}
y_i&= &G_i(q)u_i+W_i(q)s_i,\\
o_{i}& = & F_{i}(q)y_i,\\
e_i&= & r_i-y_i,
\end{array}\right.
\end{align}
where $W_i:=\operatorname{row}_{j\in \mathcal{N}_i} W_{ij}$ and $F_i:=\operatorname{col}_{j\in \mathcal{N}_i} F_{ij}$, and the time $t$ is omitted for brevity. The interconnection of system $\mathcal{P}_i$ and $\mathcal{P}_j$, $(i,j)\in \mathcal{E}$, is defined by
\begin{align} \label{eq:sioi}
s_{ij}=o_{ji} \quad \text{ and } \quad s_{ji}=o_{ij}.
\end{align}

We consider a structured reference model described by
\begin{align} \label{eq:ki}
\mathcal{K}_i:\left\{\begin{array}{lll}
y_i^d & = &T_i(q)r_i+Q_i(q)k_i,\\
p_i & = & P_i(q)y_i^d,
\end{array}\right.
\end{align}
where $Q_i:=\operatorname{row}_{j\in \mathcal{N}_i} Q_{ij}$ and $P_i:=\operatorname{col}_{j\in \mathcal{N}_i} P_{ij}$ and the interconnection variables are similarly partitioned as for $\mathcal{P}_i$, i.e., $k_i:=\operatorname{col}_{j\in \mathcal{N}_i} k_{ij}$ and $p_i:=\operatorname{col}_{j\in \mathcal{N}_i} p_{ij}$. For each pair $(i,j)\in \mathcal{E}$ the interconnection of $\mathcal{K}_i$ and $\mathcal{K}_j$ is defined by
\begin{align} \label{eq:kipi}
k_{ij}=p_{ji} \quad \text{ and } \quad k_{ji}=p_{ij}.
\end{align}
Hence, $\mathcal{K}_i$ and $\mathcal{K}_j$ can only be interconnected if $\mathcal{P}_i$ and $\mathcal{P}_j$ are interconnected. A particular case of such a reference model occurs when a decoupled closed-loop system is desired, i.e., $Q_{ij}=0$ and $P_{ij}=0$, $i,j=1,2,\dots, L$.

For the control of the interconnected system described by \eqref{eq:pi} and  \eqref{eq:sioi}, we consider that each system $\mathcal{P}_i$ is associated with a (parametrized) controller $\mathcal{C}_i$, which is a linear discrete-time system that has the tracking error $e_i$ as an input, control input $u_i$ as an output and is interconnected with other controllers $\mathcal{C}_j$ through interconnection variables $\eta_{ij}$, $\zeta_{ij}$:
\begin{align*}
\mathcal{C}_i(\rho_i):\left\{\!\!\begin{array}{ll}
u_i \!\!\!\!\!&=  C_{ii}(q,\rho_i)e_i+\displaystyle{\sum_{j\in \mathcal{N}_i}} C_{ij}(q,\rho_i)\eta_{ij},\\
\zeta_{ij} \!\!\!\!\!&=   K_{ij}(q,\rho_i)e_i+\!\!\displaystyle{\sum_{h\in \mathcal{N}_i}} K_{ijh}(q,\rho_i)\eta_{ih},\, j\in \mathcal{N}_i.
\end{array}\right.
\end{align*}
The interconnection of $\mathcal{C}_i$ and $\mathcal{C}_j$, $(i,j)\in\mathcal{E}$ is defined by
\begin{align}
\eta_{ij}=\zeta_{ji} \quad \text{ and } \quad \eta_{ji}=\zeta_{ij}.
\end{align}
By defining $\eta_i:=\operatorname{col}_{j\in \mathcal{N}_i} \eta_{ij}$ and $\zeta_i:=\operatorname{col}_{j\in \mathcal{N}_i} \zeta_{ij}$, we compactly represent controller $i$ by
\begin{align} \label{eq:ci}
\mathcal{C}_i(\rho_i):\begin{bmatrix}
u_i\\ \zeta_i
\end{bmatrix}=C_i(q,\rho_i)\begin{bmatrix}
e_i\\ \eta_i
\end{bmatrix}.
\end{align}
An example of a reference model and controlled interconnected system is provided in Figure 1 for illustration purposes.
It is assumed that each controller matrix is parametrized linearly, i.e., $[C_i(q,\rho_i)]_{ij}=\rho_i^\top \overline{C}_{ij}(q)$ for some vector of transfer functions $\overline{C}_{ij}$. The family of parametrized controllers for node $i$ is $\mathscr{C}_i:=\{C_i(q,\rho_i)\,|\, \rho_i\in\mathbb{R}^{l_i}\}$.

By stacking all the interconnection variables of the interconnected system described by \eqref{eq:pi} and \eqref{eq:sioi} as $s:=\operatorname{col} (s_1,\dots, s_L)$ and $o:=\operatorname{col}(o_1,\dots,o_L)$, we can write
\begin{align*}
y=Gu+Ws,\quad o=Fy,\quad s=\Delta o,
\end{align*}
with $G=\operatorname{diag}(G_1,\dots, G_L)$, $W=\operatorname{diag}(W_1,\dots, W_L)$, $F=\operatorname{diag}(F_1,\dots, F_L)$ and the matrix $\Delta$ defined by aggregating \eqref{eq:sioi} for all corresponding index pairs. The input-output behavior of the network is $y=(I-W\Delta F)^{-1} Gu$.
\begin{assumption}
The interconnected system and reference model satisfy $\det(I-W\Delta F)\neq 0$ and $\det(I-Q\Delta P)\neq 0$.
\end{assumption}
\vspace{-1em}
\begin{assumption}
The reference model is such that $y^d\neq r$ for all non-zero $r$, i.e., $\det((I-Q\Delta P)^{-1}T-I)\neq 0$.
\end{assumption}

\begin{problem} \label{prob:ctrproblem}
Given the parametrized controllers $\mathcal{C}_i$ and the reference models $\mathcal{K}_i$, the considered distributed controller synthesis problem is  
\begin{align}
\min_{\rho_1,\dots,\rho_L}\!\!\!J_\text{MR}(\rho_1,\dots,\rho_L)=\!\!\!\min_{\rho_1,\dots,\rho_L}\sum_{i=1}^L\bar{E}[y_i^d(t)-y_i(t)]^2\!,\label{eq:ctrproblem}
\end{align}
where $\bar{E}:=\lim_{N\rightarrow \infty} \frac{1}{N}\sum_{t=1}^N E$ and $E$ is the expectation.
\end{problem}

\begin{figure}
\centering
\centerline{\includegraphics[scale=1]{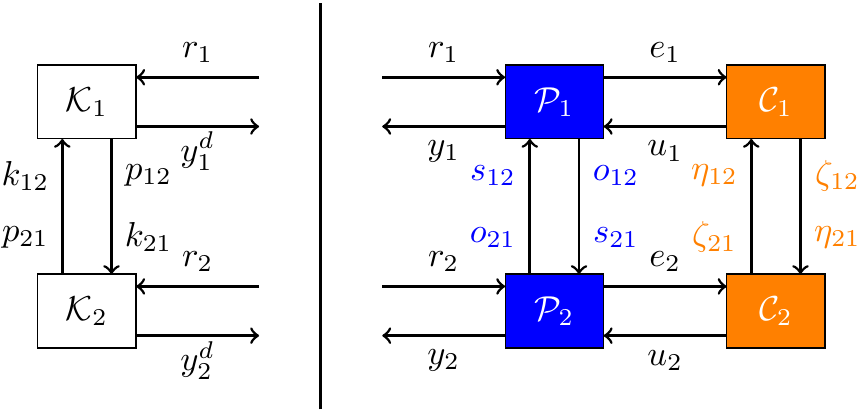}}
\caption{Structured reference model (left) and closed-loop network with distributed controller (right) for $L=2$.}
\vspace{-1em}
\label{fig:dstrbtdwmirror}
\end{figure}

\section{Ideal distributed controller synthesis} \label{sec:ideal}
A controller that admits the same structure as the interconnected system and for which the closed-loop network matches the structured reference model exactly, i.e., $y_i=y_i^d$ for all $i=1,\dots,L$, is called an \emph{ideal distributed controller}. To derive such an ideal controller, consider the interconnection of a subsystem $\mathcal{K}_i$ of the structured reference model with a subsystem $\mathcal{P}_i$ of the interconnected system, i.e.,
\begin{align} \label{eq:intkipi}
\begin{bmatrix}
y_i\\o_i\\e_i\\y_i^d\\p_i
\end{bmatrix}=\begin{bmatrix}
G_iu_i+W_is_i\\ F_iy_i\\ r_i-y_i\\T_ir_i+Q_ik_i\\ P_iy_i^d
\end{bmatrix}\quad \text{ and } \quad y_i^d=y_i.
\end{align}
Elimination of the variables $y_i^d$, $y_i$ and $r_i$ in \eqref{eq:intkipi} yields a local controller $\mathcal{C}_i^d$, described by (denoting $x_i$ by $x_i^c$ for the interconnection variables to distinguish controller variables from plant variables)
\begin{align} \label{eq:cid}
\mathcal{C}_i^d:\begin{bmatrix}
u_i\\ o_i^c\\ p_i^c
\end{bmatrix}\!\!=\!\!\underbrace{\begin{bmatrix}\dfrac{T_i}{G_i(1-T_i)} & -\dfrac{1}{G_i}W_i & \dfrac{1}{G_i(1-T_i)}Q_i\\ \dfrac{T_i}{1-T_i}F_i & 0 & \dfrac{1}{1-T_i}F_iQ_i\\ \dfrac{T_i}{1-T_i}P_i & 0 & \dfrac{1}{1-T_i}P_iQ_i
\end{bmatrix}}_{=:C_i^d(q)}\!\!
\begin{bmatrix}
e_i\\ s_i^c\\ k_i^c
\end{bmatrix}\!.\raisetag{10pt}
\end{align}
The distributed controller is constructed by interconnecting local controllers $\mathcal{C}_i^d$ and $\mathcal{C}_j^d$, $(i,j)\in\mathcal{E}$, as
\begin{align} \label{eq:sokpc}
\begin{bmatrix}
s_{ij}^c\\ k_{ij}^c
\end{bmatrix}=\begin{bmatrix}
o_{ji}^c\\ p_{ji}^c
\end{bmatrix}\quad \text{ and } \quad \begin{bmatrix}
s_{ji}^c\\ k_{ji}^c
\end{bmatrix}=\begin{bmatrix}
o_{ij}^c\\ p_{ij}^c
\end{bmatrix}
\end{align}

\begin{thm} \label{thm:idealdistrcontr}
The closed-loop network described by \eqref{eq:pi}~-~\eqref{eq:sioi} and the distributed controller \eqref{eq:cid}~-~\eqref{eq:sokpc} satisfies 
\begin{align*}
y_i=y_i^d,\qquad i=1,\dots, L.
\end{align*}
\end{thm}
\begin{proof}
Let the control variables $(u_i,e_i)$ and controller interconnection variables $(s_i^c,o_i^c,k_i^c,p_i^c)$ satisfy \eqref{eq:cid} for all $i$ and $\eqref{eq:sokpc}$ for all $(i,j)\in \mathcal{E}$, i.e., $s^c=\Delta o^c$ and $k^c=\Delta p^c$. We will first show that there exist latent variables $r_i^c:\mathbb{Z}\to \mathbb{R}$ and $y_i^c:\mathbb{Z}\to \mathbb{R}$ for each $i$, so that
\begin{align} \label{eq:dagger}
\begin{bmatrix}
y_i^c\\ o_i^c\\ e_i\\ y_i^c\\ p_i^c
\end{bmatrix}=\begin{bmatrix}
G_iu_i+W_is_i^c\\ F_iy_i^c\\ r_i^c-y_i^c\\ T_ir_i^c+Q_ik_i^c\\ P_iy_i^c
\end{bmatrix}.
\end{align}
Define $y_i^c:=G_iu_i+W_is_i$ and $r_i^c:=e_i+y_i^c$. We then have to show that $y_i^c=T_ir_i^c+Q_ik_i^c$, $o_i^c=F_iy_i^c$ and $p_i^c=P_iy_i^c$. By \eqref{eq:cid} we have that
\begin{align*}
u_i&=\frac{T_i}{G_i(1-T_i)}e_i+\frac{1}{G_i(1-T_i)}Q_ik_i^c-\frac{1}{G_i}W_is_i^c\\
&\Leftrightarrow\\
&(1-T_i)G_iu_i=T_ie_i+Q_ik_i^c-(1-T_i)W_is_i^c,
\end{align*}
which, by the definition of $y_i^c$, is equivalent with
\begin{align} \label{eq:star}
(1-T_i)y_i^c=T_ie_i+Q_ik_i^c
\end{align}
and hence, by the definition of $r_i^c$, $y_i^c=T_ir_i^c+Q_ik_i^c$.
By \eqref{eq:star} and \eqref{eq:cid}, it follows that
\begin{align*}
o_i^c&=\frac{T_i}{1-T_i}F_ie_i+\frac{1}{1-T_i}F_iQ_ik_i^c=F_iy_i^c,\\
p_i^c&=\frac{T_i}{1-T_i}P_ie_i+\frac{1}{1-T_i}P_iQ_ik_i^c=P_iy_i^c.
\end{align*}
Next, define $y^c:=\operatorname{col}(y_1^c,\dots, y_L^c)$ and $u:=\operatorname{col}(u_1,u_2,\dots,u_L)$. It follows by \eqref{eq:dagger} that $y^c=Gu+Ws^c$ and $o^c=Fy$, such that, by $s^c=\Delta o^c$, $y^c=(I-W\Delta F)^{-1}Gu$. Similarly, define $r^c:=\operatorname{col}(r_1^c,\dots, r_L^c)$ to obtain $y^c=(I-Q\Delta P)^{-1}Tr^c$ by \eqref{eq:dagger}, with $Q=\operatorname{diag}(Q_1,\dots, Q_L)$, $P=\operatorname{diag}(P_1,\dots, P_L)$ and $T=\operatorname{diag}(T_1,\dots,T_L)$. Thus, using $e=r^c-y^c$, the controller satisfies
\begin{align}
u&=G^{-1}(I-W\Delta F)(I-Q\Delta P)^{-1}T\nonumber\\
&\qquad \times(I-(I-Q\Delta P)^{-1}T)^{-1}e. \label{eq:cdext}
\end{align}
Finally, the process $y=(I-W\Delta F)^{-1}Gu$ with $e=r-y$ and the controller \eqref{eq:cdext} yield $y=(I-Q\Delta P)^{-1}Tr=y_d$, which concludes the proof.
\end{proof}

Given a stable structured reference model, the ideal distributed controller will yield a stable closed-loop network in the sense that the transfer $r\to y_d$ is stable. Regarding the internal stability, the following conditions for attaining stable ideal controllers can be obtained from \eqref{eq:cid}:
\begin{enumerate}
\item If $G_i$ has non-minimum phase zeros, then these must also be zeros of $T_i$, $W_{ij}$ and $Q_{ij}$.
\item The unstable poles of $W_{ij}$ must also be poles of $G_i$.
\item The unstable poles of $F_{ij}$ must also be zeros of $T_i$.
\end{enumerate}
Condition 1) is in agreement with the single-process case, cf.~\cite{bazanella2011}.
The other consideration is involved with local controllers being causal, i.e., that the elements of the transfer matrices $C_i^d(q)$ have a non-negative relative degree. By analyzing \eqref{eq:cid}, we observe that the relative degree of non-zero $T_i$, $Q_{ij}$ and $W_{ij}$ must be larger than or equal to the relative degree of $G_i$.

In the sequel, it will be assumed that the ideal distributed controller belongs to the parametrized class of distributed controllers. This is formalized in the following assumption, where, for ease of exposition, we introduce the permutation matrices $P_i:=\operatorname{diag}(1,\overline{P}_i)$, $i=1,\dots,L$, such that $\operatorname{col}(s_i^c,k_i^c)=\overline{P}_i\operatorname{col}_{j\in \mathcal{N}_i} \operatorname{col}(s_{ij}^c,k_{ij}^c)$.

\begin{assumption}\label{as:cinc}
$P_i^\top C_i^d P_i\in\mathscr{C}_i$ for each $i=1,\dots,L$.
\end{assumption}
We associate $\rho_1^d,\dots, \rho_L^d$ with the ideal distributed controller, such that $P_i^\top C_i^dP_i=C_i(\rho_i^d)$ for $i=1,\dots, L$. By Theorem~\ref{thm:idealdistrcontr}, $(\rho_1^d,\dots,\rho_L^d)$ solves problem~\eqref{eq:ctrproblem}. The following simple example briefly illustrates the ideal distributed controller constructed in this section.
\begin{example} \label{xmpl:two}
Consider two coupled processes
\begin{align*}
y_1(t)&=G_1(q)u_1(t)+G_{12}(q)y_2(t),\\
y_2(t)&=G_2(q)u_2(t)+G_{21}(q)y_1(t),
\end{align*}
with transfer functions
\begin{align*}
G_1(q)&=\frac{c_1}{q-a_1},\quad G_{12}(q)=\frac{d_1}{q-a_1},\\
G_2(q)&=\frac{c_2}{q-a_2},\quad G_{21}(q)=\frac{d_2}{q-a_2}.
\end{align*}

The objective is that the closed-loop interconnected system behaves as two decoupled processes with first-order dynamics, according to
\begin{align} 
y_i^d(t)=T_i(q)r_i(t),\quad T_i(q)=\frac{1-\gamma_i}{q-\gamma_i},\quad i=1,2.\label{eq:T2}
\end{align}

Now, via \eqref{eq:cid}, we find that the ideal distributed controller is described by
\begin{align*}
\begin{bmatrix}
u_1\\o_1^c
\end{bmatrix}=\begin{bmatrix}
C_{11}^d & C_{12}^d\\ K_{12}^d & 0
\end{bmatrix}\begin{bmatrix}
e_1\\ s_1^c
\end{bmatrix}\text{ and }\begin{bmatrix}
u_2\\o_2^c
\end{bmatrix}=\begin{bmatrix}
C_{22}^d & C_{21}^d\\ K_{21}^d & 0
\end{bmatrix}\begin{bmatrix}
e_2\\ s_2^c
\end{bmatrix}
\end{align*}
with $e_i=r_i-y_i$, the interconnections $s_1^c=o_2^c$, $s_2^c=o_1^c$, and
\begin{align*}
C_{ii}^d(q) &= \frac{1-\gamma_i}{c_1}\frac{q-a_i}{q-1},\quad C_{ij}^d(q)=-\frac{d_i}{c_i},\\
 K_{ij}^d(q)&=\frac{1-\gamma_i}{q-1},\quad (i,j)\in \mathcal{E}.
\end{align*}
\end{example}

\section{Data-driven distributed controller} \label{sec:data}
In the remainder of this paper, we shall consider the case where $F_{ij}=1$ for all $i=1,\dots,L$, $j\in \mathcal{N}_i$. This implies that all systems are directly coupled through the outputs $y_i$:
\begin{align} \label{eq:net}
y_i =G_i(q)u_i+\sum_{j\in \mathcal{N}_i}W_{ij}(q)y_j,\quad i=1,\dots, L,
\end{align}
compactly written as $y=W_Iy+Gu$, with $[W_I]_{ij}=W_{ij}$. The interconnected system can always be represented as in \eqref{eq:net} without changing the transfer $u\to y$, by replacing $W_{ij}$ in \eqref{eq:net} by $\bar{W}_{ij}=W_{ij}F_{ji}$, since $s_{ij}=F_{ji}(q)y_j$ by \eqref{eq:pi}-\eqref{eq:sioi}.

The controller described in Section~\ref{sec:ideal} provides a solution to Problem~\ref{prob:ctrproblem}, but requires $\mathcal{P}_i$ to be given. The problem considered in this section, is the direct data-driven synthesis: given data $\{u_i,\,y_i\}$, $i=1,\dots, L$, solve problem~\eqref{eq:ctrproblem}.

We address this problem by two steps: virtual reference generation and distributed controller identification.

\subsection{Virtual reference generation}
Consider data $\{u_i,\,y_i\}$, $i=1,\dots, L$ collected from the network \eqref{eq:net}. This data can be obtained in closed loop with a stabilizing controller or in open loop if the network is stable, i.e., if $(I-W_I)^{-1}G$ is stable. For the reference model described by \eqref{eq:ki}-\eqref{eq:kipi}, we recall that $y_d=(I-Q\Delta P)^{-1}Tr$. Now, given $y_1,y_2,\dots, y_L$, consider the computation of the \emph{virtual reference} signals $\bar{r}_1,\bar{r}_2,\dots, \bar{r}_L$ according to the structured reference model as
\begin{align} \label{eq:rbar}
y=(I-Q\Delta P)^{-1}T
\bar{r}.
\end{align}
Then $\bar{r}_1,\bar{r}_2,\dots, \bar{r}_L$ are such that, when the network \eqref{eq:net} is in closed loop with the ideal distributed controller, fictitiously, the measured outputs $y_1,y_2,\dots,y_L$ are the corresponding outputs. Solving \eqref{eq:rbar} requires the data $y_1,y_2,\dots, y_L$ to be collected by a central governor. Because central data collection is not favourable, we propose to generate the virtual reference signals locally. This can always be done for the considered reference model, by determining the virtual reference signals $\bar{r}_1,\bar{r}_2,\dots, \bar{r}_L$ and the \emph{virtual interconnection} signals $\bar{p}_1,\bar{p}_2,\dots, \bar{p}_L$ according to \eqref{eq:ki} and \eqref{eq:kipi} so that
\begin{align*}
y_i=T_i\bar{r}_i+\sum_{j\in \mathcal{N}_i} Q_{ij}\bar{p}_{ji}
\quad \text{ and }\quad \bar{p}_{ij}=P_{ij}y_i,\quad j\in \mathcal{N}_i.
\end{align*}
Given a virtual reference signal $\bar{r}_i$, the corresponding virtual tracking error and, hence, the input to the ideal controller, is $\bar{e}_i=\bar{r}_i-y_i$. The virtual reference generation can thus be distributed, as summarized in Algorithm~\ref{alg:ref}.

 \begin{algorithm}
 \caption{Distributed virtual reference computation} \label{alg:ref}
 \begin{algorithmic}[1]
 \renewcommand{\algorithmicrequire}{\textbf{Input:}}
 \renewcommand{\algorithmicensure}{\textbf{Output:}}
 \REQUIRE Reference model transfer functions $T_i$, $Q_i$, $P_i$ and output data $y_i$ for $i=1,\dots, L$
 \ENSURE  Virtual signals $\bar{r}_i$, $\bar{e}_i$, $\bar{p}_i$ for $i=1,\dots, L$
  \FOR {$i = 1$ to $L$}
  \STATE Compute $\bar{p}_{i}$ such that $\bar{p}_{i}(t)=P_{i}(q)y_i(t)$.
  \ENDFOR
  \FOR {$i=1 $ to $L$}
  \STATE Receive $\bar{p}_{ji}$ from nodes $j\in \mathcal{N}_i$. Compute $\bar{r}_i$ such that
  \begin{align*}
  T_i(q)\bar{r}_i(t)=y_i(t)-\sum_{j\in \mathcal{N}_i} Q_{ij}(q)\bar{p}_{ji}(t).
  \end{align*}
  \STATE $\bar{e}_i\leftarrow\bar{r}_i-y_i$
  \ENDFOR
 \RETURN $\bar{r}_i$, $\bar{e}_i$, $\bar{p}_i$, $i=1,\dots, L$
 \end{algorithmic} 
 \end{algorithm}
 
 \begin{figure}[!t]
\centering
\includegraphics[scale = .72]{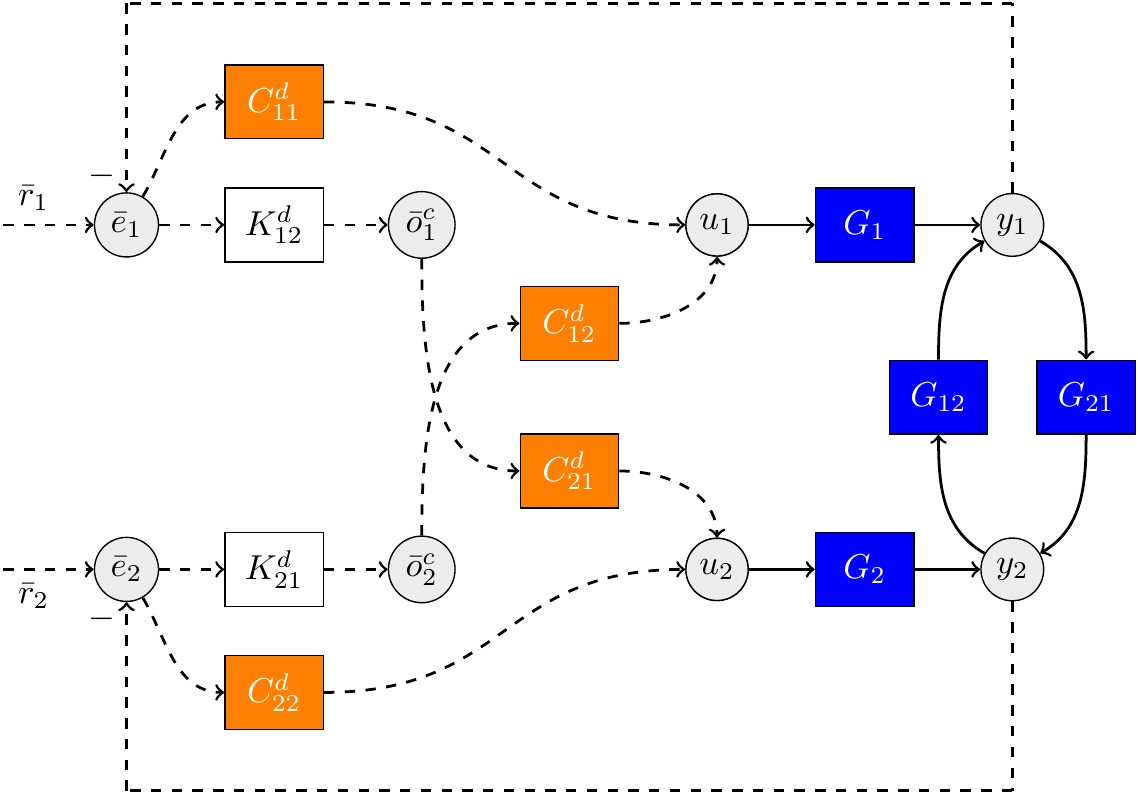}
\caption{Virtual experiment setup for identification of the ideal distributed controller in Example~\ref{xmpl:two}.}
\label{fig:virtexp}
\end{figure}

\subsection{Distributed controller identification}
Let us return to Example~\ref{xmpl:two}. Figure~\ref{fig:virtexp} shows the constructed virtual network that is obtained by following Algorithm~\ref{alg:ref}. The task of determining the controllers $\mathcal{C}_1^d$ and $\mathcal{C}_2^d$ now essentially becomes a dynamic network identification problem \cite{vandenhofetal2013}, where $\{C_{11}^d,C_{12}^d,K_{12}^d\}$ and $\{C_{22}^d,C_{21}^d,K_{21}^d\}$ are the modules to be identified (strictly speaking $\{C_{11}^d,C_{12}^d\}$ and $\{C_{22}^d,C_{21}^d\}$, since $K_{12}^d$ and $K_{21}^d$ are known). The signals $u_1$ and $u_2$ are directly available from the measurements, while $\bar{e}_1$ and $\bar{e}_2$ are virtual and obtained by Algorithm~\ref{alg:ref}. The virtual controller interconnection signals $\bar{o}_{ij}^c$ are obtained by filtering $\bar{e}_i$ as $\bar{o}_{12}^c=K_{12}^d\bar{e}_1$ and $\bar{o}_{21}^c=K_{21}^d\bar{e}_2$.

To illustrate the identification, consider the parametrized models $\{C_{11}(\rho_1),C_{12}(\rho_1)\}$, $\{C_{22}(\rho_2),C_{21}(\rho_2)\}$, the predictors
\begin{align} \label{eq:pred}
\hat{u}_1(\rho_1)&=C_{12}(\rho_1)\bar{o}_{21}^c+C_{11}(\rho_1)\bar{e}_1,\\
\hat{u}_2(\rho_2)&=C_{21}(\rho_2)\bar{o}_{12}^c+C_{22}(\rho_2)\bar{e}_2
\end{align}
and the identification criterion
\begin{align*}
J_\text{VR}(\rho_1,\rho_2)=\bar{E}[\varepsilon_1(\rho_1)]^2+\bar{E}[\varepsilon_2(\rho_2)]^2
\end{align*}
with $\varepsilon_i:=u_i-\hat{u}_i(\rho_i)$. We will now analyse the minima of $J_\text{VR}$. Since $\bar{o}_{12}^c=K_{12}^d\bar{e}_1$ and $\bar{o}_{21}^c=K_{21}^d\bar{e}_2$, it follows that
\begin{align*}
\varepsilon_1(\rho_1)&=(C_1^d-C_1(\rho_1))\bar{e}_1+(C_{12}^d-C_{12}(\rho_1))K_{21}^d \bar{e}_2,\\
\varepsilon_2(\rho_2)&=(C_2^d-C_2(\rho_1))\bar{e}_2+(C_{21}^d-C_{21}(\rho_2))K_{12}^d \bar{e}_1.
\end{align*}
Then, since $\bar{e}=(T^{-1}-I)(I-W_I)^{-1}Gu$, where $T=\operatorname{diag}(T_1,T_2)$, the prediction errors are
\begin{align*}
\begin{bmatrix}
\varepsilon_1(\rho_1)\\ \varepsilon_2(\rho_2)
\end{bmatrix}&=\begin{bmatrix}
C_{11}^d-C_{11}(\rho_1) & (C_{12}^d-C_{12}(\rho_1))K_{21}^d\\ (C_{21}^d-C_{21}(\rho_2))K_{12}^d & C_{22}^d-C_{22}(\rho_2)
\end{bmatrix}\\
&\quad \times(T^{-1}-I)(I-W_I)^{-1}Gu.
\end{align*}
It now appears that a global minimum of $J_{\text{VR}}$ is $(\rho_1,\rho_2)=(\rho_1^d,\rho_2^d)$ and that this minimum is unique if the control input signal $u=\operatorname{col}(u_1,u_2)$ from the experiment is persistently exciting of a sufficient order. Hence, the global minimum of $J_\text{VR}(\rho_1,\rho_2)$ is then the same as the global minimum of $J_\text{MR}(\rho_1,\rho_2)$, where $J_\text{VR}$ is quadratic in $\rho$ when the models are parametrized linearly in $\rho$. The distributed-controller synthesis problem is therefore reformulated as a network identification problem.

The latter reasoning for Example~\ref{xmpl:two} leads us to the following result for a general interconnected system:
\begin{thm}\label{thm:id}
Consider the predictor $\hat{u}_i(\rho_i):=C_{ii}(\rho_i)\bar{e}_i+\sum_{j\in \mathcal{N}_i} C_{ij}^W(\rho_i)\bar{o
}_{ji}^c+C_{ij}^Q(\rho_i)\bar{p}_{ji}$ with $\bar{o}_{ji}^c=(1-T_j)^{-1}T_j\bar{e}_j+\sum_{h\in \mathcal{N}_j}(1-T_j)^{-1}Q_{jh}\bar{p}_{hj}$. The identification criterion
\begin{align*}
J_i^\text{VR}(\rho_i)=\bar{E}[u_i-\hat{u}_i(\rho_i)]^2
\end{align*}
has a global minimum point at $\rho_i^d$ and this minimum is unique if the spectrum of $w_i=\operatorname{col}(\bar{e}_i, \operatorname{col}_{j\in \mathcal{N}_i} \bar{o}_{ji}^c, \operatorname{col}_{j\in \mathcal{N}_i} \bar{p}_{ji})$, denoted $\Phi_{w_i}(\omega)$, is positive definite for all $\omega\in[-\pi,\pi]$.
\end{thm}
\begin{proof}
First, we note that $\bar{p}_{ji}=p_{ji}^c$ and $\bar{o}_{ji}^c=o_{ji}^c$, where $p_{ji}^c$ and $o_{ji}^c$ satisfy \eqref{eq:cid} and \eqref{eq:sokpc} for $e_i=\bar{e}_i$, $i=1,\dots,L$. Consequently, by Corollary~1 in \cite{vandenhofetal2013}, it follows that $\rho_i^d$ is  the unique global minimum point of $J_\text{VR}$.
\end{proof}
When the reference model is decoupled, the spectrum condition can be translated directly to the spectrum of the input:
\begin{corollary} \label{cor:id}
Let $P_i=0$, $Q_i=0$ and consider the predictors $\hat{u}_i^D(\rho_i):=C_{ii}(\rho_i)\bar{e}_i+\sum_{j\in \mathcal{N}_i} C_{ij}(\rho_i)\bar{o}_{ji}^c$, $i=1,2,\dots, L$. The identification criterion
\begin{align*}
J_\text{VR}(\rho_1,\dots,\rho_L)=\sum_{i=1}^L \bar{E}[u_i-\hat{u}_i^D(\rho_i)]^2
\end{align*}
has a global minimum point at $(\rho_1^d,\dots, \rho_L^d)$ and this minimum is unique if $\Phi_u(\omega)$ is positive definite for all $\omega\in[-\pi,\pi]$.
\end{corollary}

The condition on $\Phi_u$ in Corollary~\ref{cor:id} can be realized by appropriate experiment design. The condition on $\Phi_{w_i}$ in Theorem~\ref{thm:id}, however, cannot always be realized by an appropriate design of $u$. For instance, consider Example~\ref{xmpl:two}, but now with non-zero $Q_i,P_i$. Then the number of entries of $w_1=\operatorname{col}(\bar{e}_1,\bar{o}_{21},\bar{p}_{21})$ is larger than the number of inputs in $u=\operatorname{col}(u_1,u_2)$, hence $\Phi_{w_1}(\omega)$ cannot be positive definite. We observe that the excitation condition can be relaxed if we do not require $\rho_i^d$ to be the only global minimum of $J_i^\text{VR}$.
\begin{corollary}
Each global minimum point $\rho_i^*$ of $J_i^\text{VR}$ satisfies $C_{ii}(\rho_i^*)=C_{ii}(\rho_i^d)$ and for all $j\in\mathcal{N}_i$:
\begin{align}\label{eq:condeq}
(C_{ij}^W(\rho_i^*)-C_{ij}^W(\rho_i^d))+(C_{ij}^Q(\rho_i^*)-C_{ij}^Q(\rho_i^d)P_{ji}=0
\end{align}
if $\Phi_{\xi_i}(\omega)$, $\xi_i=\operatorname{col}(\bar{e}_i, \operatorname{col}_{j\in \mathcal{N}_i} \bar{o}_{ji}^c)$, is positive definite for all $\omega\in[-\pi,\pi]$.
\end{corollary}
It can verified that $(\rho_1^*,\dots,\rho_L^*)$, satisfying $C_{ii}(\rho_i^*)=C_{ii}(\rho_i^d)$ and \eqref{eq:condeq} for all $(i,j)\in\mathcal{E}$, is also a global minimum point of $J_{\text{MR}}$ and hence solves problem \eqref{eq:ctrproblem}.

Each identification criterion $J_i^\text{VR}(\rho_i)$, $i=1,\dots,L$, can be minimized separately. The required predictor inputs for node $i$ are the virtual signals obtained in Algorithm~\ref{alg:ref}, which are available locally ($\bar{e}_i$) or communicated by nodes $j\in \mathcal{N}_i$ ($\bar{o}_{ji}^c$ and $\bar{p}_{ji}$). Of course, $J_i^\text{VR}$ is not to be considered in practice since it involves expectations; for a finite number ($N$) of data, a solution $\rho_i$ is obtained by minimizing $\bar{J}_i^\text{VR}(\rho_i)=\sum_{t=1}^N (u_i(t)-\hat{u}_i(t,\rho_i))^2$. Observe the following difference with respect to multi-variable VRFT \cite{campestrini2016}. Instead of identifying the transfer from $\bar{e}$ to $u$, we identify the local controller dynamics $C_i^d$ by exploiting the structure of the interconnected system.

\section{Illustrative example: 9-systems network}
Consider the interconnected system \eqref{eq:net} with $L=9$ and the interconnection structure depicted in Figure~\ref{fig:P}. The transfer functions describing the dynamics are of order one and given by
\begin{align*}
G_i&=\frac{1}{q-a_i},\quad W_{ij}=\frac{0.1}{q-a_i},\quad i=1,\dots,9,
\end{align*}
with $a_i\in(0,1)$. It is desired to decouple the interconnected system and to have the same step response for every output channel. Hence the reference model is chosen as $y_i^d=T_i^d(q)r_i$, where
\begin{align*}
T_i(q)=\frac{0.4}{q-0.6},\quad i=1,\dots,9.
\end{align*}

We collect the data $\{u_i(t),y_i(t), t=1,2,\dots,100\}$ from \eqref{eq:net} in open-loop, with mutually uncorrelated white-noise input signals $u_i$ having a standard deviation of $\sigma_{u_i}=1$. Hence, we are in the situation of Corollary~\ref{cor:id}. Each controller $C_i$, $i=1,\dots, 9$, is parametrized such that Assumption~\ref{as:cinc} holds. Since there is no noise present in the output, the optimization of $\bar{J}_i^\text{VR}$ with predictors \eqref{eq:pred} yields the parameters $\rho_i^d$ and therefore $J_\text{MR}$ is equal to zero.

\begin{figure}[!t]
      \centering
      \subcaptionbox{Interconnected system\label{fig:P}}
        [.49\linewidth]{\includegraphics[scale=0.85]{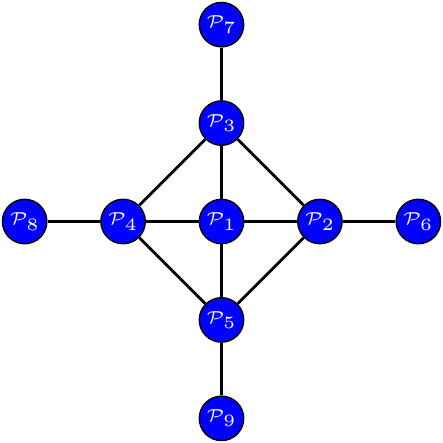}}
      \subcaptionbox{Distributed controller\label{fig:C}}
        [.49\linewidth]{\includegraphics[scale=0.85]{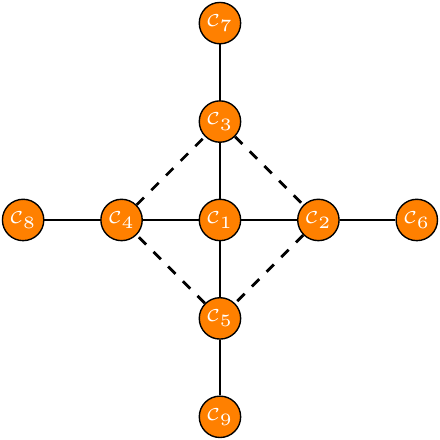}}
      \caption{Graph for the interconnected system (a) and distributed controller (b) that satisfies Assumption~\ref{as:cinc} (solid and dashed edges), that has reduced number of links (solid edges) and that is decentralized (no edges).}\label{fig:PC}\vspace*{-1em}
    \end{figure}

Next, we will analyze the situation where noise affects the system, by considering disturbed outputs $\tilde{y}_i(t)=y_i(t)+v_i(t)$ for the synthesis, with $v_i$ white-noise processes with standard deviations $\sigma_{v_i}=0.1$ that are mutually uncorrelated and uncorrelated with $u_i$. The method of generating virtual references and predictors is kept the same. The resulting distributed controller is interconnected with the plant and a step reference is applied to each subsystem simultaneously, with an amplitude between zero and one. Figure~\ref{fig:timeseriesy} shows the output response of the closed-loop network in red together with the response of the reference model (in black) on the left. We observe only a minor difference between the responses, due to the noise added to the data for identification, as shown in Figure~\ref{fig:timeseriesy} on the right.

\begin{figure}[!t]
\centering
\includegraphics[width=3.5in]{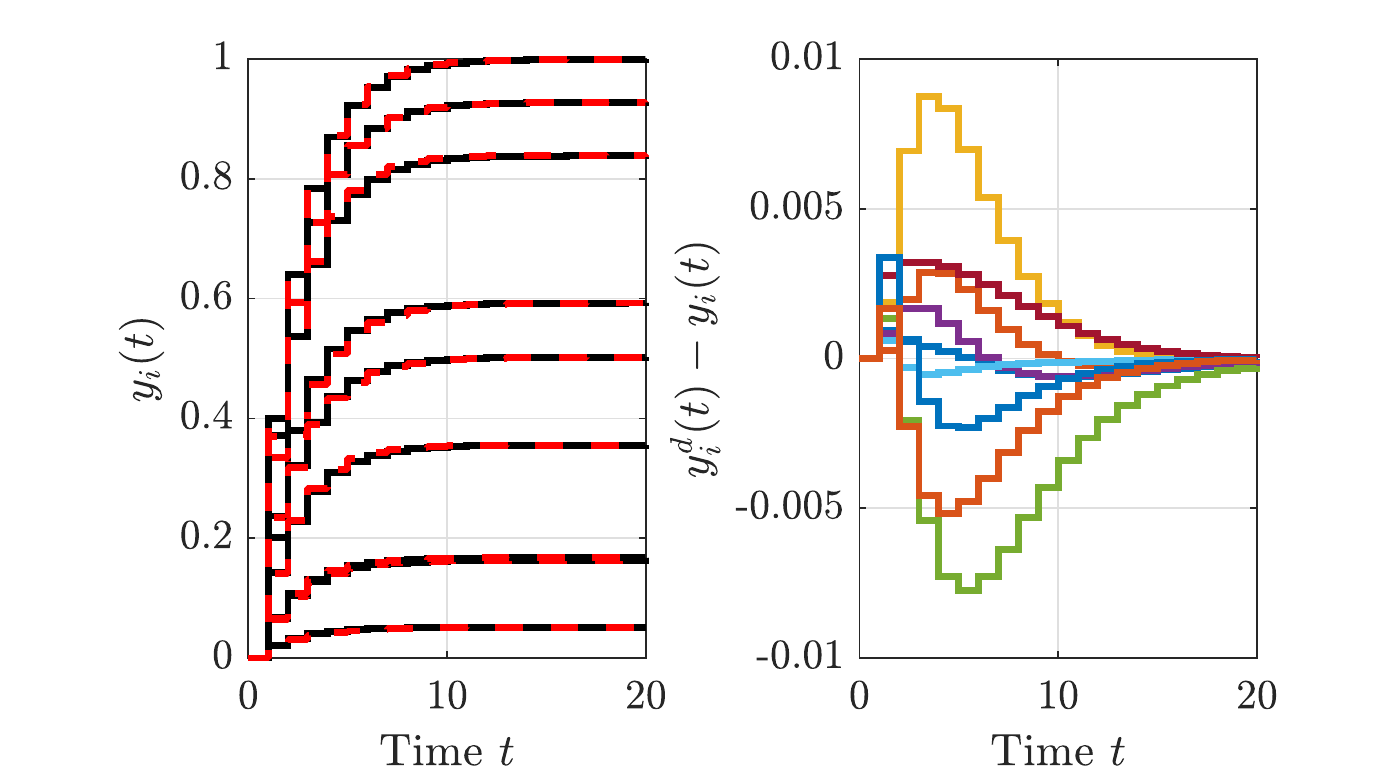}
\caption{Left: Closed-loop response of the network with the data-driven distributed controller (red) and the desired response of the structured reference model (black). Right: Response errors with respect to the desired outputs.}
\label{fig:timeseriesy} \vspace{-1em}
\end{figure}

The distribution of the error between the achieved closed-loop network with the identified distributed controller and the structured reference model resulting from 100 Monte Carlo runs is presented in Figure~\ref{fig:perfcomp}. Because $P_i^\top C_i^d P_i\in\mathscr{C}_i$ for all $i$, the error between the achieved closed-loop system and the reference model is only due to the noise. Assumption~\ref{as:cinc} does not always hold in practice. To illustrate such a situation, consider the case where four communication links are not present, represented in Figure~\ref{fig:C} by the dashed edges. The implication is that in the parametrization $C_{ij}(\rho_i)=0$ for the corresponding edges $(i,j)\in \mathcal{E}$ and, e.g., $\hat{u}_2(\rho_2)=C_{22}\bar{e}_2+C_{21}(\rho_2)\bar{o}_{12}^c+C_{26}(\rho_2)\bar{o}_{62}^c$. Note that links are thus not removed \emph{ a posteriori}, but the interconnection structure for the distributed controller is induced by the controller class. As shown by Figure~\ref{fig:perfcomp}, there is a significant performance degradation, because the controller class is not `rich' enough, although the graph for the controller remains connected. We finally consider the data-driven synthesis of a decentralized controller, corresponding to $C_{ij}(\rho_i)=0$ for all $(i,j)\in \mathcal{E}$. The resulting discrepancy between reference model and closed-loop network is plotted in Figure~\ref{fig:perfcomp} and shows a further decrease in performance.

\begin{figure}[!t]
\centering
\includegraphics[width=3.5in]{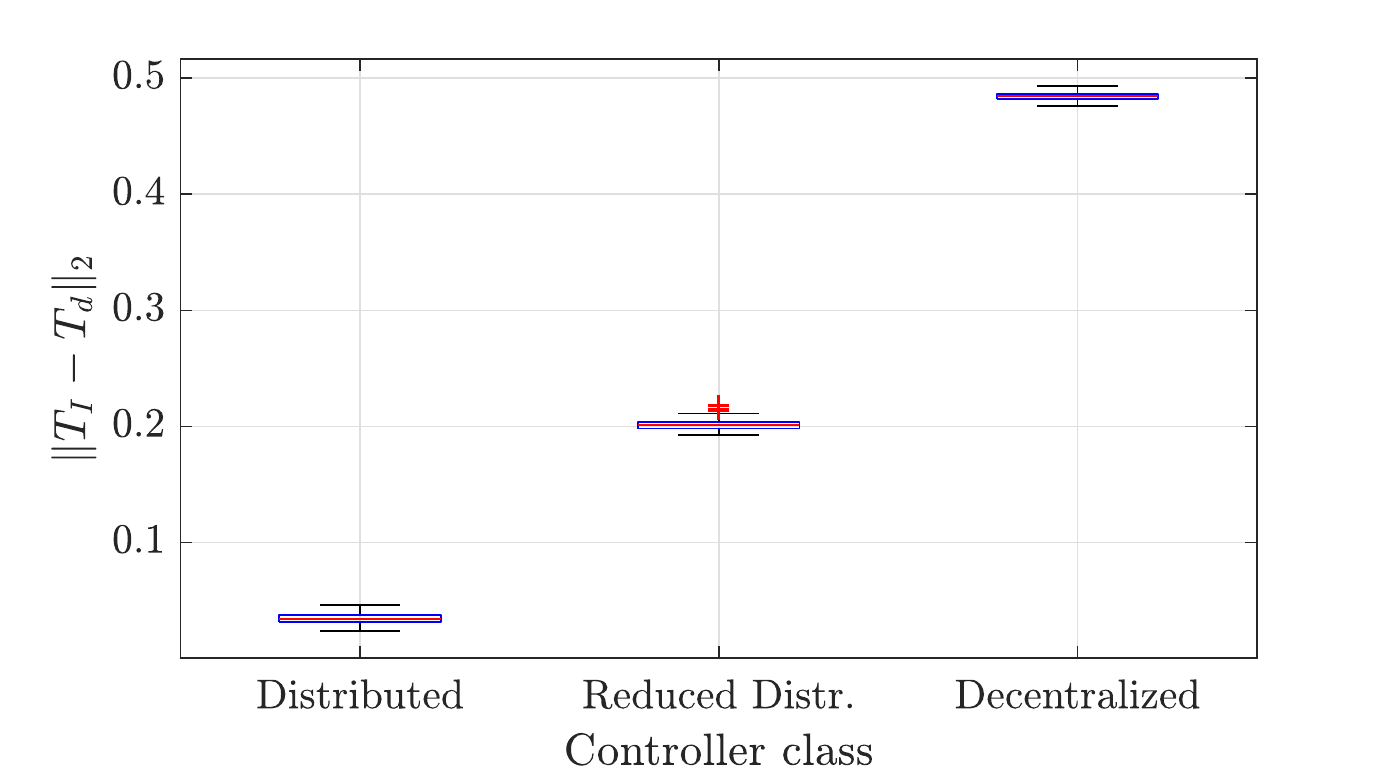}
\caption{Distribution of the achieved performance for various controller classes, where $T_I$ and $T_d$ denote the transfers $r\to y$ and $r\to y_d$, respectively.}
\label{fig:perfcomp}\vspace{-1em}
\end{figure}

\section{Concluding remarks}
We have presented a method for data-driven synthesis of distributed controllers for interconnected systems defined on arbitrary graphs. The method follows from two main steps: (i) the development of an ideal distributed controller class and (ii) the data-driven modelling of this controller via the identification of modules in a (virtual) dynamic network. Both the synthesis (identification) and implementation of the constructed data-driven distributed controllers can be distributed, which improves scalability compared to the standard multi-variable VRFT method.

In future work, we will consider distributed controller design in the presence of (process) noise sources in the network and for non-ideal controller classes.

\bibliographystyle{IEEEtran}
\bibliography{./rfrncs20}

\begin{thebibliography}{10}
\providecommand{\url}[1]{#1}
\csname url@samestyle\endcsname
\providecommand{\newblock}{\relax}
\providecommand{\bibinfo}[2]{#2}
\providecommand{\BIBentrySTDinterwordspacing}{\spaceskip=0pt\relax}
\providecommand{\BIBentryALTinterwordstretchfactor}{4}
\providecommand{\BIBentryALTinterwordspacing}{\spaceskip=\fontdimen2\font plus
\BIBentryALTinterwordstretchfactor\fontdimen3\font minus
  \fontdimen4\font\relax}
\providecommand{\BIBforeignlanguage}[2]{{%
\expandafter\ifx\csname l@#1\endcsname\relax
\typeout{** WARNING: IEEEtran.bst: No hyphenation pattern has been}%
\typeout{** loaded for the language `#1'. Using the pattern for}%
\typeout{** the default language instead.}%
\else
\language=\csname l@#1\endcsname
\fi
#2}}
\providecommand{\BIBdecl}{\relax}
\BIBdecl

\bibitem{reviewarticle2017}
F.~Lamnabhi-Lagarrigue, A.~Annaswamy, S.~Engell, A.~Isaksson, P.~Khargonekar,
  R.~M. Murray, H.~Nijmeijer, T.~Samad, D.~Tilbury, and P.~{Van den Hof},
  ``Systems \& control for the future of humanity, research agenda: Current and
  future roles, impact and grand challenges,'' \emph{Annual Reviews in
  Control}, vol.~43, pp. 1 -- 64, 2017.

\bibitem{christofides2013}
P.~D. Christofides, R.~Scattolini, D.~{Mu\~{n}oz de la Pe\~{n}a}, and J.~Liu,
  ``Distributed model predictive control: A tutorial review and future research
  directions,'' \emph{Computers \& Chemical Engineering}, vol.~51, pp. 21 --
  41, 2013.

\bibitem{chen2019}
X.~Chen, H.~Xu, and M.~Feng, ``${H_2}$ performance analysis and ${H_2}$
  distributed control design for systems interconnected over an arbitrary
  graph,'' \emph{Systems \& Control Letters}, vol. 124, pp. 1 -- 11, 2019.

\bibitem{langbortetal2004}
C.~Langbort, R.~S. Chandra, and R.~D'Andrea, ``Distributed control design for
  systems interconnected over an arbitrary graph,'' \emph{IEEE Transactions on
  Automatic Control}, vol.~49, no.~9, pp. 1502--1519, 2004.

\bibitem{vandenhofetal1995}
P.~M.~J. {Van den Hof} and R.~J.~P. {Schrama}, ``Identification and control:
  Closed-loop issues,'' \emph{Automatica}, vol.~31, no.~12, pp. 1751 -- 1770,
  1995.

\bibitem{vandenhofetal2013}
P.~M.~J. {Van den Hof}, A.~G. {Dankers}, P.~S.~C. {Heuberger}, and X.~Bombois,
  ``Identification of dynamic models in complex networks with prediction error
  methods -- {B}asic methods for consistent module estimates,''
  \emph{Automatica}, vol.~49, no.~10, pp. 2994 -- 3006, 2013.

\bibitem{hou2013}
Z.-S. Hou and Z.~Wang, ``From model-based control to data-driven control:
  Survey, classification and perspective,'' \emph{Information Sciences}, vol.
  235, pp. 3 -- 35, 2013.

\bibitem{bazanella2011}
A.~Bazanella, L.~Campestrini, and D.~Eckhard, \emph{Data-Driven Controller
  Design: The $H_2$ Approach}, ser. Communications and Control
  Engineering.\hskip 1em plus 0.5em minus 0.4em\relax Springer Netherlands,
  2011.

\bibitem{campi2002}
M.~Campi, A.~Lecchini, and S.~Savaresi, ``Virtual reference feedback tuning: a
  direct method for the design of feedback controllers,'' \emph{Automatica},
  vol.~38, no.~8, pp. 1337 -- 1346, 2002.

\bibitem{campestrini2016}
L.~Campestrini, D.~Eckhard, L.~A. Chia, and E.~Boeira, ``Unbiased {MIMO} {VRFT}
  with application to process control,'' \emph{Journal of Process Control},
  vol.~39, pp. 35 -- 49, 2016.

\bibitem{huff2019}
D.~D. Huff, L.~Campestrini, G.~R. Gon{\c c}alves~da Silva, and A.~S. Bazanella,
  ``Data-driven control design by prediction error identification for
  multivariable systems,'' \emph{Journal of Control, Automation and Electrical
  Systems}, vol.~30, no.~4, pp. 465--478, 2019.

\bibitem{formentin2015}
S.~Formentin, A.~Bisoffi, and T.~Oomen, ``Asymptotically exact direct
  data-driven multivariable controller tuning,'' \emph{IFAC-PapersOnLine},
  vol.~48, no.~28, pp. 1349 -- 1354, 2015, 17th IFAC Symposium on System
  Identification SYSID 2015.

\end{thebibliography}

\end{document}